\documentclass{amsart}
\usepackage[letterpaper,hmargin=1.2in,vmargin=1.88in]{geometry}

\usepackage{amsmath, amsthm, amssymb}
\usepackage{algorithm}
\usepackage{algpseudocode}
\algnewcommand\algorithmicinput{\textbf{Input:}}
\algnewcommand\INPUT{\item[\algorithmicinput]}
\algnewcommand\algorithmicoutput{\textbf{Output:}}
\algnewcommand\OUTPUT{\item[\algorithmicoutput]}
\algnewcommand\algorithmicoracle{\textbf{Oracle:}}
\algnewcommand\ORACLE{\item[\algorithmicoracle]}
\makeatletter
\algnewcommand{\LineComment}[1]{\Statex \hskip\ALG@thistlm \(\triangleright\) #1}
\makeatother

\usepackage{url}

\usepackage{mathtools}
\usepackage{tikz}
\usetikzlibrary{arrows, positioning, shadows}

\newtheorem{theorem}{Theorem}[section]
\newtheorem{lemma}[theorem]{Lemma}

\theoremstyle{definition}

\newtheorem{problem}[theorem]{Problem}

\theoremstyle{remark}
\newtheorem{remark}[theorem]{Remark}

\makeatletter
\def\paragraph{\@startsection{paragraph}{4}%
  \z@\z@{-\fontdimen2\font}%
  {\normalfont\bfseries}}
\makeatother

\numberwithin{equation}{section}

\newcommand{\Z}{\mathbb Z}
\newcommand{\F}{\mathbb F}

\usepackage{tikz}
\usetikzlibrary{matrix,arrows}

\title[Problems Equivalent to Beating Exponent 3/2 for Polynomial Factorization]{Algebraic Problems Equivalent to Beating Exponent 3/2 for Polynomial Factorization over Finite Fields}

\author{Zeyu Guo}
\author{Anand Kumar Narayanan}
\author{Chris Umans}
\address{Department of Computing and Mathematical Sciences, California Institute of Technology.}
\email{zguo,anandkn,umans@caltech.edu}
\thanks{The authors were supported by NSF grant CCF 1423544 and a Simons Foundation Investigator grant.}
\begin{document}

\begin{abstract}
The fastest known algorithm for factoring univariate polynomials over
finite fields is the Kedlaya-Umans \cite{ku} (fast modular
composition) implementation of the Kaltofen-Shoup algorithm
\cite[\S~2]{ks}. It is randomized and takes 
$\widetilde{O}(n^{3/2}\log q + n \log^2 q)$ time to factor polynomials of degree $n$ over the finite field $\mathbb{F}_q$ with $q$ elements. A significant open problem is if the $3/2$ exponent can be improved. We study a collection of algebraic problems and establish a web of reductions between them. A consequence is that an algorithm for any one of these problems with exponent better than $3/2$ would yield an algorithm for polynomial factorization with exponent better than $3/2$.
\end{abstract}

\maketitle

\section{Introduction}
A recent trend in discrete algorithms has been to establish very
efficient reductions between problems with polynomial time algorithms, with the intention of identifying barriers (conceptual or concrete) to improving the polynomial running time of the best known
algorithms. A standard example is the problem 3-SUM, which seems to
require essentially quadratic time, and which has been reduced to many other problems. More recently, the study of ``fine-grained'' complexity has broadened, with several connections established between central problems in discrete algorithms, and new conjectures beyond the 3-SUM conjecture entering the picture (see, e.g. \cite{agw,aw,aw2,pat,rz,wy,wil,ww}).\\ \\
In this paper we focus on a ``barrier'' in {\em algebraic} algorithms, that of improving the exponent $3/2$ for univariate polynomial factorization and several other problems. Generally, algebraic problems have two relevant ``size'' parameters -- $n$, and the field size $q$. It is typical for the dependence on $q$ to be
polylogarithmic (it is for all of the problems we consider), and so we focus on the exponent on $n$ in this work. We find that exponent $3/2$ seems to be a barrier for a number of problems. This points to a need to move beyond the so-called ``baby steps giant steps''
methodology which tends to give rise to the exponent 3/2 behavior. \\ \\
The reductions in this paper can be seen as giving evidence that
improving the 3/2 exponent may not be possible for these problems, but we believe that it ``merely'' gives evidence that this improvement requires a conceptual breakthrough (along the lines of going beyond the baby-steps giant-steps approach). Using the connections established in this paper, such a breakthrough for any one of the problems considered here would improve the exponent for all of them. \\ \\
In the discussion below, we use $\widetilde{O}$ to suppress $n^{o(1)}$
terms and $\log^{o(1)}q$ terms, in order to highlight the exponent on
$n$ that is our main object of study. We also use the phrase ``nearly
linear time reduction'' to mean a reduction that runs in time
$\widetilde{O}(n\log q)$, and the phrase ``3/2 exponent reducible'' to mean the weaker connection that shows that beating exponent 3/2 for one problem implies beating exponent 3/2 for the other. 
\subsection{Algebraic problems with 3/2 exponent algorithms}

We investigate the complexity of factoring a univariate polynomial over a finite field into its irreducible factors. The problem formally stated is,
\begin{itemize}
 \item \textsc{Factor:} \textit{Given a monic square free $f(x) \in \mathbb{F}_q[x]$ of degree $n$, write $f(x)$ as a product of its monic irreducible factors.}
\end{itemize}
The square free assumption is without loss of generality \cite{knu,yun}. \textsc{Factor} can be solved in randomized polynomial time \cite{ber} and there is an extensive line of research
\cite{cz,ks,gs} leading to a randomized algorithm \cite{ku} with exponent $3/2$. Surprisingly, even determining the {\em degree} of a single irreducible factor rapidly would be sufficient to improve the exponent of this algorithm. We formulate this problem as
\begin{itemize}
\item \textsc{Factor Degree:} \textit{Given a monic square free $f(x) \in \mathbb{F}_q[x]$, find the degree of an irreducible factor of $f(x)$.} 
\end{itemize}
and prove in \S~\ref{section_factor_degree} that \textsc{Factor} is
$3/2$-exponent reducible to \textsc{Factor Degree}. That is, an
algorithm for \textsc{Factor Degree} with exponent less than $3/2$
yields one for \textsc{Factor}. Observe that \textsc{Factor Degree}
merely seeks one, not necessarily all, irreducible factor degrees. We next investigate two linear algebraic problems, both we will
demonstrate to be nearly linear time reducible to \textsc{Factor}.
\begin{itemize}
\item \textsc{Frobenius Min-Poly:} \textit{Given a monic square free $f(x) \in \mathbb{F}_q[x]$, compute the minimal polynomial of the Frobenius endomorphism on $\mathbb{F}_q[x]/(f(x))$ which takes $a(x) \mod f(x)$ to $a(x)^q \mod f(x)$.}
\item \textsc{Carlitz Char-Poly:} \textit{Given a monic square free $f(x) \in \mathbb{F}_q[x]$, compute the characteristic polynomial of the Carlitz endomorphism on $\mathbb{F}_q[x]/(f(x))$ which takes $a(x) \mod f(x)$ to $xa(x) + a(x)^q \mod f(x)$.}
\end{itemize}
In \S~\ref{section_carlitz}, we prove that \textsc{Factor Degree} is nearly linear time reducible to \textsc{Carlitz Char-Poly}, under certain restrictions on the characteristic of $\F_q$. These restrictions were removed in \cite{nar} by passing from Carlitz to Drinfeld modules. In \S~\ref{section_minimal}, through a novel recursive argument, we prove that \textsc{Factor} is $3/2$-exponent reducible to \textsc{Frobenius Min-Poly}.\\ \\
\textsc{Frobenius Min-Poly} was known \cite{ks,kl} to be nearly linear time reducible to 
\begin{itemize}
\item \textsc{Automorphism Projection:} \textit{Given a monic square free $f(x) \in \mathbb{F}_q[x]$, $\alpha\in \mathbb{F}_q[x]/(f(x))$ and an $\mathbb{F}_q$-linear map $u:\mathbb{F}_q[x]/(f(x)) \longrightarrow \mathbb{F}_q$, compute $u(\alpha^{q^i}), \forall i \in \{1,2,\ldots,\deg(f)\}$.}
\end{itemize}
Thus, as a consequence of the reduction in \S~\ref{section_minimal}, we conclude that \textsc{Factor} is $3/2$-exponent
reducible to \textsc{Automorphism Projection}. This should be
contrasted with the connection established in \cite{ks,kl}.  They show that \textsc{Factor} is nearly linear time
reducible to \textsc{Automorphism Projection} {\em assuming an
  $\mathbb{F}_q$-linear straight line program algorithm} for
\textsc{Automorphism Projection}. This assumption allows them to use
the ``transpose'' problem \textsc{Automorphism Evaluation}. Our
reduction to \textsc{Automorphism Projection} is novel, direct, and
 holds without any assumptions. \\ \\ 
The final two problems pertain to zero testing Moore and Vandermonde determinants.
\begin{itemize}
\item \textsc{Moore-Det:} \textit{Given a monic square free $f(x) \in \mathbb{F}_q[x]$ and a positive integer $m$, decide if the determinant of the $m$ by $m$ square matrix with entries $m_{ij}:=x^{j{q^i}} \mod f(x)$ is zero}.
\item \textsc{Vandermonde-Det:} \textit{Given a monic square free $f(x) \in \mathbb{F}_q[x]$ and a positive integer $b \leq \sqrt{\deg(f)}$, decide if the determinant of the Vandermonde matrix with first row $$(x^{q^i} \mod f(x), i = 0,1,2,\ldots,b-1,b,2b,3b,\ldots,(b-1)b,b^2)$$ is zero}.
\end{itemize}
In \S~\ref{section_moore}, we prove that \textsc{Factor Degree} is nearly linear time reducible to each of these problems and that each of these problems is nearly linear time reducible to \textsc{Factor}. In summary, we have the following diagram where solid lines denote nearly linear time reductions and dotted lines denote $3/2$-exponent reductions.
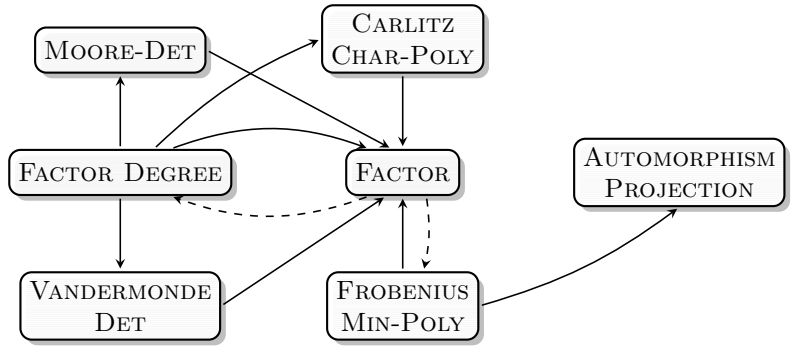
\begin{center}
\begin{tikzpicture}[node distance=10mm and 15mm,
every node/.style={rectangle, minimum size=6mm, rounded corners=1.5mm, thick ,draw=black,  top color=white, drop shadow,  align=center},
ar/.style={->,>=stealth, semithick, shorten >= 1pt, shorten <= 1pt},
dashed_ar/.style={->,>=stealth,dashed, semithick, shorten >= 1pt, shorten <= 1pt},
font=\scshape
]
\node (facdeg) {Factor Degree};
\node (fac) [right=of facdeg] {Factor};
\node (carlitz) [above=of fac] {Carlitz\\Char-Poly};
\node (frob) [below=of fac] {Frobenius \\Min-Poly};
\node (moore) [above=of facdeg] {Moore-Det};
\node (vander) [below=of facdeg] {Vandermonde\\Det};
\node (autproj) [right=of fac] {Automorphism\\Projection};
\draw [ar] (facdeg.35) to [ bend left=10] (carlitz.west);
\draw [ar] (facdeg.25) to [bend left=20]  (fac.145);
\draw [dashed_ar] (fac.315)  to [bend left=10] (frob.59);
\draw [ar] (carlitz.south) to (fac.north);
\draw [ar] (frob.north) to (fac.south);
\draw [dashed_ar] (fac.215) to [bend left=20]  (facdeg.335);
\draw [ar] (moore.east) to (fac.115);
\draw [ar] (frob.east) to [ bend right=10] (autproj.south);
\draw [ar] (facdeg.north) to (moore.south);
\draw [ar] (facdeg.south) to (vander.north);
\draw [ar] (vander.east) to (fac.235);
\end{tikzpicture}
\end{center}
An interesting open question is if the dotted lines can be made
solid. Except for \textsc{Automorphism Projection}, every listed
problem has a known randomized algorithm with exponent $3/2$. If the
matrix multiplication exponent is $2$, then a randomized algorithm for
\textsc{Automorphism Projection} with exponent $3/2$ is known. Another
open problem is if this dependence on the matrix multiplication
exponent can be removed -- perhaps by reducing \textsc{automorphism
  projection} to one of the other problems in the figure. Regardless, an algorithm
for any of the problems in the figure with exponent less that $3/2$
would yield an algorithm with exponent $3/2$ for \textsc{Factor}, and
this is one of the main points of this paper. 
\section{Factorization and Finding a Factor Degree}\label{section_factor_degree}

Clearly, if one can solve the problem \textsc{factor} in time $T(n,
q)$ then one can solve the problem \textsc{factor degree} in time
$T(n, q)$. In this section we show a reduction in the reverse
direction, which leads to the surprising conclusion that one  {\em only} needs to compute the {\em degree} of a single irreducible factor of the polynomial $f(x)$ with exponent better than 3/2 to be able to factor $f(x)$ completely with exponent better than 3/2.

\begin{theorem}
If there is an algorithm that solves \textsc{factor degree} in the time $T(n,q)$ where $T(n,q)=\Omega(n\log^2 q)$ \footnote{The assumption $T(n,q)=\Omega(n\log^2 q)$ is without loss of generality. For otherwise we slow down an algorithm with runtime $T(n,q)$ until it is $\Omega(n\log^2 q)$.}, then there is an algorithm that solves \textsc{factor} in time $\widetilde{O}(n\cdot T(n,q)^{1/3} \log^{4/3} q)$.
\end{theorem}

Observe that when \textsc{factor degree} has an exponent 3/2
algorithm (as it does), this reduction recovers a 3/2 exponent
algorithm for \textsc{factor}. A sub-3/2 exponent algorithm for
\textsc{factor degree} implies a sub-3/2 exponent algorithm for
\textsc{factor}, with a nearly-linear time algorithm yielding exponent $4/3$ for \textsc{factor}.

\begin{proof}
We are given a monic, square-free polynomial $f(x) \in \F_q[x]$ of
degree $n$. Let $g(x)$ be the product of irreducible factors of $f(x)$
with degrees at most $t$ (for a parameter $t$ to be chosen later). If
$s(x)$ is defined as
\[s(x) = \prod_{i=1}^t(x^{q^i} - x)^{a_i},\]
for some positive integers $a_1, a_2, \ldots, a_t$, then we have that $g(x) = \gcd(s(x), f(x))$. Using fast modular
composition \cite{ku} and the method of Kaltofen-Shoup \cite{ks}, we
can compute $s(x) \bmod f(x)$ in time $\widetilde{O}(n\sqrt{t}
\log^2 q)$ time. We then proceed to factor $g(x)$ completely, using
the Kedlaya-Umans implementation of the Kaltofen-Shoup algorithm. The
bottleneck in this algorithm is computing the splitting polynomials,
which are all polynomials of the form of $s(x)$, with $i$ ranging from
$1$ up to $t' \le t$. This portion of the algorithm runs in time
$\widetilde{O}(n\sqrt{t}\log^2 q)$ and factors $g(x)$ completely.

Now we invoke the algorithm to solve \textsc{factor degree}, on input
$f(x)/g(x)$. Upon finding the degree $d$ of an irreducible factor, we
compute $\gcd(x^{q^d} - x \bmod f(x), f(x))$ to split off the factors
with that degree. We then repeat. The number of repetitions is bounded by $n/t$,
since each irreducible factor of $f(x)/g(x)$ has degree at least
$t$. 
Each repetition takes time $T(n,q)+\widetilde{O}(n\log^2 q)$.
Thus this portion of the algorithm runs in time
$\widetilde{O}(n/t\cdot (T(n,q)+n\log^2 q))=\widetilde{O}(n/t \cdot T(n,q))$. 
Finally we factor completely using equal-degree
factorization which takes $\widetilde{O}(n\log^2 q)$ time. Optimizing, we set $t = (T(n,q)/\log^2 q)^{2/3}$, and the overall running time
becomes
\[
\widetilde{O}(n\cdot T(n,q)^{1/3} \log^{4/3} q)
\]
  for each of the two stages, and hence in total as well.
\end{proof}

\section{Factoring and Minimal Polynomial of Frobenius}\label{section_minimal}
For a monic square free $f(x)$, let $g(\lambda) \in \F_q[\lambda]$ denote the minimal polynomial of the $q^{th}$ power Frobenius endomorphism $\sigma: \F_q[x]/(f(x)) \rightarrow \F_q[x]/(f(x))$. That is, $g(\lambda)$ is the unique nonzero monic polynomial of least degree such that the endomorphism $g(\sigma)$ on $\F_q[x]/(f(x))$ is zero. The problem \textsc{Frobenius Min-Poly} is to determine $g(\lambda)$ given $f(x)$. Since $g(\lambda)$ is the least common multiple of $\lambda^d-1$ as $d$ runs through the degrees of the irreducible factors of $f(x)$, \textsc{Frobenius Min-Poly} is nearly linear time reducible to \textsc{Factor}. In this section, we conversely prove that \textsc{Factor} is $3/2$-exponent reducible to \textsc{Frobenius Min-Poly}. 

Let $\mathtt{FrobMinPoly}$ be an oracle that solves \textsc{Frobenius Min-Poly}. We present an algorithm $\mathtt{Factor}$ that invokes $\mathtt{FrobMinPoly}$ and solves \textsc{Factor}.
For $k\in\mathbb{N}^+$, denote by $\Phi_k$ the $k$th cyclotomic polynomial over $\F_q$. Write $\phi(\cdot)$ for the Euler totient function.
\begin{algorithm}[H]
\caption{$\mathtt{Factor}(f(x))$}\label{algo2}
\begin{algorithmic}[1]
\INPUT  Monic square free polynomial $f(x) \in \F_q[x]$ of degree $n$.
\OUTPUT Monic irreducible factors of $f(x)$.
\ORACLE $\mathtt{FrobMinPoly}$
\State Using \cite{ku}, output and remove all monic irreducible factors of $f(x)$ of degree at most $n^{2/3}$.  If at most one irreducible factor of degree greater than $n^{2/3}$ remains, output and exit.
\State $g(\lambda) \gets \mathtt{FrobMinPoly}(f(x))$.
\State Perform square free factorization on $g(\lambda)$, and then run $\mathtt{Factor}$ recursively on the outputs to obtain the list of monic irreducible factors $g_1(\lambda),\dots,g_m(\lambda)$ of $g(\lambda)$.
\State Run $\mathtt{FindT}(g_1(\lambda),\dots,g_m(\lambda))$ computing the set $T:=\{k: p\nmid  k \text{ and } \Phi_k(\lambda) | g(\lambda)\}$ as well as $m_k$, the multiplicity of $\Phi_k(\lambda)$ in $g(\lambda)$, for each $k\in T$.
\State Compute $S:=\{kp^e: k\in T, 0\leq e\leq \log_p m_k\}$. 
\State \textbf{for each} $s \in S$ greater than $n^{2/3}$, set $f_s(x) \gets \gcd(f(x), x^{q^s}-x \bmod f(x))$, $f(x)\gets f(x)/f_s(x)$ and perform equal-degree factorization on $f_s(x)$.
\end{algorithmic}
\end{algorithm}
The algorithm begins by extracting all monic irreducible factors of degree at most $n^{2/3}$. After Line 1, $f(x)$ only has large (at least $n^{2/3}$) degree factors. Suppose $d_1,d_2,\ldots,d_m$ are the degrees of the (remaining) monic irreducible factors of $f(x)$. Then the minimal polynomial $g(\lambda) \in \F_q[\lambda]$ of the Frobenius acting on $\F_q[x]/(f(x))$ is
$$g(\lambda)=\mathrm{lcm}\left(\lambda^{d_1}-1,\dots, \lambda^{d_m}-1\right).$$
In particular, the cyclotomic polynomials $\Phi_{d_1}(\lambda),\dots, \Phi_{d_m}(\lambda)$ divide $g(\lambda)$ and the factorization of $g(\lambda)$ contains information about $d_1,d_2,\ldots,d_m$. We devise a novel procedure to infer $d_1,d_2,\ldots,d_m$ efficiently.\\ \\
On Line 2, $g(\lambda)$ is computed by invoking $\mathtt{FrobMinPoly}$. To infer $d_1,d_2,\ldots,d_m$, we seek the factorization of $g(\lambda)$. To this end, a key idea is to factor $g(\lambda)$ recursively on Line $3$ and obtain a list $g_1(\lambda),g_2(\lambda),\ldots,g_m(\lambda)$ of its monic irreducible factors. Since $f(x)$ is not irreducible at this point, $g(\lambda)$ has degree strictly less than $f(x)$ and the algorithm runs to completion.

Then we use a procedure $\mathtt{FindT}(g_1(\lambda),\dots,g_m(\lambda))$ to compute the set $T$ and integers $m_k$ as defined on Line $4$. This step is the most technical part of the algorithm, and we defer its description and analysis to the next subsection, where we prove the following theorem:
\begin{theorem}\label{thm_timeT}
$\mathtt{FindT}(g_1(\lambda),\dots,g_m(\lambda))$  can be implemented to run in $\widetilde{O}(n\log q)$ time.  
\end{theorem}
Once $T$ is known, to compute $S$ on Line $5$ is straightforward. 
The following lemma shows that $S$ indeed contains $d_1,d_2,\ldots,d_m$.
\begin{lemma}$d_1,d_2,\ldots,d_m \in S$.
\end{lemma}
\begin{proof}
Consider an arbitrary $d \in \{d_1,d_2,\ldots,d_m\}$ and write it as $d=kp^e$ with $p\nmid k$. By definition $(\lambda^{d}-1) | g(\lambda)$. Since $\lambda^d-1=(\lambda^k-1)^{p^e}$ and $\lambda^k-1=\prod_{k_0|k} \Phi_{k_0}(\lambda)$, $\Phi_k(\lambda)$ is a factor of $g(\lambda)$ with multiplicity at least $p^e$. So $k \in T$ and $m_k\geq p^e$, implying $d=kp^e \in S$.
\end{proof}
To conclude, by Line $6$, all the irreducible factors of $f(x)$ are indeed output.

\begin{theorem}
Suppose the oracle $\mathtt{FrobMinPoly}$ runs in time $T(n,q)$ which is monotone in $n$ and $q$. Then  $\mathtt{Factor}$ factors a degree-$n$ polynomial in $ \widetilde{O}(T(n,q)+ n^{4/3}\log^{2} q)$ time.
\end{theorem}
\begin{proof}
We first analyze the running time of each step except the recursive call.  Line $1$ can be implemented in $\widetilde{O}(n^{4/3} \log^{2} q)$ time using the baby-step-giant-step strategy \cite{ks,ku}. The oracle $\mathtt{FrobMinPoly}$ on Line 2 runs in time $T(n,q)$. The set $T$ on Line 4 could be found in time $\widetilde{O}(n \log q)$ by Theorem \ref{thm_timeT}. Since $\prod_{k\in T}\Phi_k(\lambda)^{m_k}$ divides $g(\lambda)$, we have $\sum_{k\in T} m_k \phi(k)\leq \deg(g(\lambda))\leq n$. Hence $|T|\leq n$ and $m_k\leq n$ for all $k \in T$, implying $S$ on Line 5 could be computed in time $\widetilde{O}(n)$.
Further,
$$
\sum_{s\in S} s\leq \sum_{k\in T, 0\leq e\leq \log_p m_k} k p^e\leq \log n \sum_{k\in T} k m_k \leq O(\log\log n) \cdot \log n \sum_{k\in T} m_k\phi(k)=\widetilde{O}(n)
$$
where we use $k/\phi(k)=O(\log\log k)$ \cite{rs} and $\sum_{k\in T} m_k\phi(k)\leq n$. 
Hence the number of $s\in S$ greater than $n^{2/3}$ is at most $(\sum_{s\in S} s)/n^{2/3}=\widetilde{O}(n^{1/3})$. For each $s\in S$,  Computing $f_s(x)$  takes $\widetilde{O}(n\log^{2} q)$ time for each $s\in S$ \cite{ku} and hence $\widetilde{O}(n^{4/3}\log^{2} q)$ time in total.
Equal degree factorization on Line 6 takes $\widetilde{O}(n \log^{2} q)$ time in total. \\ \\
Let $d_\mathrm{max}(f(x))$ denote the maximal degree of the irreducible factors of $f(x)$. We claim that $d_\mathrm{max}(f(x))$ shrinks by at least a factor of two every two recursive calls.  It implies that the recursive tree has depth no more than $O(\log n)$, so the total running time is bounded by $O(\log n)\cdot (T(n,q)+ \widetilde{O}(n^{4/3}\log^{2} q))= \widetilde{O}(T(n,q)+ n^{4/3}\log^{2} q)$, as desired.\\ \\
Consider an irreducible factor $g_0(\lambda)$ of $g(\lambda)$. We know $g_0(\lambda)$ divides $\lambda^k-1=\prod_{k_0|k} \Phi_{k_0}(\lambda)$ for a positive integer $k$ corresponding to some degree $k$ irreducible factor $f_0(x)$ of $f(x)$. If $g_0(\lambda)$ divides $\Phi_{k_0}(\lambda)$ for some proper divisor $k_0$ of $k$,  we have 
$\deg(g_0(\lambda))\leq \phi(k_0)\leq k_0\leq k/2$. Likewise, if $g_0(\lambda)$ is a proper irreducible factor of $\Phi_k(\lambda)$, we have $\deg(g_0(\lambda))\leq \phi(k)/2\leq k/2$ as well. So assume $g_0(\lambda)=\Phi_k(\lambda)$. Suppose $k=\prod_\ell \ell^{e_\ell}$, $\ell$ running over prime divisors of $k$.
Then $\phi(k)=\prod_\ell (\ell-1)\ell^{e_\ell-1}$. If $k$ is even, we have $e_2\geq 1$ implying $\deg(g_0(\lambda))=\phi(k)\leq k/2$ (since for $\ell=2$, $(\ell-1)\ell^{e_\ell-1}=\ell^{e_\ell}/2$). If $k$ is odd,
$\deg(g_0(\lambda))=\phi(k)=\prod_\ell (\ell-1)\ell^{e_\ell-1}$ is even. The argument above applied to $g_0(\lambda)$ and $g(\lambda)$ in place of $f_0(x)$ and $f(x)$ shows that the degree shrinks by at least a factor of two in the next recursive call. The claim follows.
\end{proof}
\begin{remark} One may easily check that the same algorithm and analysis also work if the polynomial $g(\lambda)$ computed by the oracle is the {\em characteristic polynomial} of the Frobenius endomorphism instead of the minimal polynomial. The only difference is that $g(\lambda)$ is the product of $\lambda^{d_1}-1,\dots, \lambda^{d_m}-1$ rather than their lcm.
\end{remark}
%\subsection{Computing the Set \texorpdfstring{$T$}{T}}
\subsection{Computing the Set $T$}
We next devise a nearly linear time procedure to implement $\mathtt{FindT}$. It relies on solutions to the following two problems: (1) finding all irreducible factors of $\Phi_k(\lambda)$  over $\F_q$ from a single irreducible factor $g_0(\lambda)$ and (2) finding the corresponding integer $k$. We deal with these two problems individually before describing $\mathtt{FindT}$. 
%\subsubsection{Finding the irreducible factors of \texorpdfstring{$\Phi_k(\lambda)$}{phi k lambda}}
\subsubsection{Finding the irreducible factors of $\Phi_k(\lambda)$}
Let $k\in [1, n]$  be an integer coprime to $p$. Our goal is to find all the irreducible factors of $\Phi_k(\lambda)$ over $\F_q$ from a single irreducible factor $g_0(\lambda)|\Phi_k(\lambda)$. To achieve it, we need to know how $\Phi_k(\lambda)$ factorizes over $\F_q$.\\
\paragraph{Factorization of $\Phi_k(\lambda)$ over $\F_q$:} 
As $k$ is coprime to $p$, there are $\phi(k)$ distinct primitive $k$th roots of unity in $\overline{\F}_q$ which are exactly the roots of  $\Phi_k(\lambda)$. Denote this set of roots by $\mu_k$.
Let $G$ be the abelian group $(\Z/k\Z)^\times$ of order $\phi(k)$.
For $d\in\Z$, we write $\bar{d}$ for the image of $d$ in $\Z/k\Z$.
The group $G$ acts on $\mu_k$ such that $\bar{d}\in G$ sends  any $\theta\in \mu_k$ to $\theta^d$.
This is a regular action, meaning that for fixed $\theta\in\mu_k$, the map $\bar{d}\mapsto \theta^d$ is a bijection between $G$ and $\mu_k$. 
As $p$ is coprime to $k$, we have $\bar{q}\in G$.
Let $G_0=\langle \bar{q}\rangle\subseteq G$ and $s=[G:G_0]$.
Restrict the $G$-action on $\mu_k$ to a $G_0$-action. Then $\mu_k$ is partitioned into $s$ distinct $G_0$-orbits represented by $\theta_1,\dots,\theta_s\in \mu_k$.
It is well-known that the factorization of $\Phi_k(\lambda)$ over $\F_q$ is then determined in the following way:
\begin{lemma}\label{lem_facphi}
Under the notations above,  $\Phi_k(\lambda)$ has $s$ irreducible factors $g_1(\lambda),\dots,g_s(\lambda)$  over $\F_q$ corresponding to the $G_0$-orbits $G_0\theta_1, \dots, G_0\theta_s$ of $\mu_k$  in the sense that the set of roots of $g_i(\lambda)$ is exactly $G_0\theta_i$.
\end{lemma}

\begin{proof}
Let $g(\lambda)$ be an irreducible factor of $\Phi_k(\lambda)$ over $\F_q$ and $\theta\in\mu_k$ be a root of $g(\lambda)$. Then $\F_q[\theta]$ is Galois over $\F_q$ with the Galois group generated by the Frobenius map $a\mapsto a^q$. So $a\in \F_q[\theta]$ is a root of $g(\lambda)$ if and only if $a^q$ is a root of $g(\lambda)$. Therefore $G_0\theta$ is the set of roots of $g(\lambda)$ and the lemma follows. 
\end{proof}

From now on we fix a root $\theta\in\mu_k$ of the given irreducible factor $g_0$ of $\Phi_k$.
For any subgroup $H\subseteq G$ containing $G_0$, the $G$-action on $\mu_k$ restricts to an $H$-action. The $H$-orbit $H\theta$ is partitioned into a disjoint union of $G_0$-orbits and hence corresponds to a subset $L$ of irreducible factors of $\Phi_k(\lambda)$ by Lemma \ref{lem_facphi}.
Note that  $L$ also determines $H$: $h\in G$ lies in $H$ if and only if the minimal polynomial of $h \theta$ over $\F_q$ is in $L$.
We say $L$ is {\em associated with} the subgroup $H$.\\

We use the following procedure $\mathtt{FindOrder}(\ell, L)$ to find the order of $H\bar{\ell}$ in $G/H$:

\begin{algorithm}[H]
\caption{$\mathtt{FindOrder}(\ell, L)$}
\begin{algorithmic}[1]
\INPUT Integer $\ell\in [1,n]$ and  $L$ associated with some subgroup $H$ containing $G_0$
\OUTPUT The order of $H\bar{\ell}$ in $G/H$, or zero if $\bar{\ell}\not\in G=(\Z/k\Z)^\times$
\State Pick arbitrary $f_0(\lambda)\in L$
\State $e\gets 0$,  $r_0\gets \lambda \bmod f_0(\lambda) \in \F_q[\lambda]/(f_0(\lambda))$ 
\Repeat
\State $e\gets e+1$
\State $r_e\gets r_{e-1}^\ell$ and let $f_e(\lambda)$ be the minimal polynomial of $r_e$ over $\F_q$
\Until{$f_e(\lambda)\in L$ or $f_e(\lambda)=f_{e'}(\lambda)$ for some $0\leq e'<e$}
\State\textbf{if} $f_e(\lambda)\in L$ \textbf{then} \Return $e$ \textbf{else} \Return $0$
\end{algorithmic}
\end{algorithm}

\begin{lemma} \label{lem_findorder}
There exists a procedure   $\mathtt{FindOrder}(\ell, L)$ that takes an integer $\ell$ and the set $L$ associated with $H$, and returns the following result:  if  $\bar{\ell}\in G=(\Z/k\Z)^\times$, it returns the order of $H\bar{\ell}$ in $G/H$, i.e. the smallest $e>1$ for which $\bar{\ell}^e \in H$.  Otherwise it returns zero. Moreover $\mathtt{FindOrder}(\ell, L)$ could be implemented in time $\widetilde{O}(\phi(k)\log q)$.
\end{lemma}

\begin{proof}
The procedure $\mathtt{FindOrder}(\ell, L)$ first picks a polynomial $f_0(\lambda)\in L$ and one of its roots $r_0\in H\theta$.
First assume $\bar{\ell}\in G$. Then the  procedure finds  the smallest $e$ for which
$f_e(\lambda)\in L$ or $f_e(\lambda)=f_{e'}(\lambda)$ for some $0\leq e'<e$, where $f_e(\lambda)$ is the minimal polynomial of $\bar{\ell}^e r_0=r_0^{\ell^e}$.
 The former condition $f_e(\lambda)\in L$ is equivalent to $G_0 \bar{\ell}^e r_0\subseteq H \theta=H r_0$, or equivalently $\bar{\ell}^e \in H$.
 The latter condition $f_e(\lambda)=f_{e'}(\lambda)$ for some $0\leq e'<e$ is equivalent to $G_0\bar{\ell}^e r_0= G_0\bar{\ell}^{e'} r_0$, or equivalently $\bar{\ell}^{e-e'}\in G_0$. Note that if the latter condition is met, by minimality of $e$ we must have $e'=0$ and hence $\bar{\ell}\in G_0\subseteq H$. So the former condition subsumes the latter and the desired $e$ is picked.\\ \\
Now assume $\bar{\ell}\not\in G$, then $\ell_0:=\mathrm{gcd}(\ell,k)>0$. Then all $r_e$ generated in  $\mathtt{FindOrder}(\ell, L)$ are $(k/\ell_0)$th roots of unity for $e>0$, and hence $f_e(\lambda)\not\in L$ for $e>0$. Let $k_1$ be the largest divisor of $k$ coprime to $\ell$ and $k_2=k/k_1$. Then $k_1$ consists of prime divisors of $k$ not appeared in the factorization of $\ell_0$ whereas $k_2$ consists of those appeared.
For any prime number $t|k_2$, we have $v_t(\ell)\geq v_t(\ell_0)\geq 1$ and $v_t(k_2)\leq \log_t k_2\leq \log k_2$, where $v_t(n)$ denotes the integer $u\geq 0$ such that $t^u|n$ and $t^{u+1}\nmid n$. So for $e\geq \log k_2$, we have $v_t(\ell^e)\geq v_t(k_2)$ for all prime number $t$, and hence $k_2|\ell^e$. On the other hand, $k_1$ is coprime to $\ell$. So $r_e=r_0^{\ell^e}$ is a primitive $k_1$th root of unity for $e\geq \log k_2$.
Let $\overline{G}_0$ be the subgroup of $\overline{G}:=(\Z/k_1\Z)^\times$ generated by $q \bmod k_1$.
We apply Lemma \ref{lem_facphi} on $\overline{G}$ and $\overline{G}_0$ instead of $G$ and $G_0$.
Let $e'=[\overline{G}: \overline{G}_0]\leq [G:G_0]$. Then for any $e\geq\log k$ we have $\overline{G}_0r_{e+e'}=\overline{G}_0\bar{\ell}^{e'}r_e=\overline{G}_0r_e$ and hence $f_{e+e'}(\lambda)=f_e(\lambda)$.
So the loop in $\mathtt{FindOrder}(\ell, L)$ is executed at most $\log k +e'\leq \log k+[G:G_0]$ times. And as $f_e(\lambda)\not\in L$ for $e>0$, it returns zero.\\ \\
As all $\F_q(r_e)$ are subfields of $\F_q(r_0)$, the degrees of all $f_e(\lambda)$ are bounded by $\mathrm{deg}(f_0(\lambda))=|G_0|$. Line 5 of $\mathtt{FindOrder}(\ell, L)$ could be computed in time $\widetilde{O}(|G_0|\log q)$ using  the Kedlaya-Umans \cite{ku} implementation of Shoup's algorithm \cite{sho}.
The condition on Line 6 could be checked in time $O(|G_0|\log q(\log N + \log |L|))$ if we store $L$ and the list of $f_i(\lambda)$ using a data structure supporting fast search and insertion, where $N$ is number of times that the loop is executed.
Here $|L|=[H:G_0]\leq k$. If $\bar{\ell}\in G$, the loop is executed $e=[H\langle\bar{\ell}\rangle: H]\leq [G:G_0]$ times, whereas if  $\bar{\ell}\not\in G$, it is executed no more than $\log k+ [G:G_0]$ times.
So the total running time is bounded by $(\log k+ [G:G_0])\cdot (\widetilde{O}(|G_0|\log q)+O(|G_0|\log q\log k ))=\widetilde{O}(\phi(k)\log q)$.
\end{proof}

We use a randomized procedure $\mathtt{FindCyclotomic}(g_0(\lambda), n)$ to find all irreducible factors of $\Phi_k(\lambda)$ over $\F_q$. Here  $g_0(\lambda)$ is one irreducible factor of $\Phi_k(\lambda)$ and $n$ is the degree of the polynomial $f(x)$.\footnote{The argument $n$ is only used on Line 2 and 3 to control the number of repetitions and the range of $\ell$, which is related to the error probability.}

\begin{algorithm}[htb]
\caption{$\mathtt{FindCyclotomic}(g_0(\lambda), n)$}\label{alg_cyclo}
\begin{algorithmic}[1]
\INPUT Irreducible factor $g_0(\lambda)$ of $\Phi_k(\lambda)$ over $\F_q$ and  degree $n$ of $f(x)$
\OUTPUT The list of irreducible factors of  $\Phi_k(\lambda)$ over $\F_q$
\State $L\gets\{g_0(\lambda)\}$
\For{$t$ \textbf{from} $1$ \textbf{to} $N=\lfloor c \log n\log\log n\rfloor$}\Comment{{\em $c>0$ is a large enough constant}}
\State Pick an integer $\ell\in [1,n]$ at random
\State $e\gets\mathtt{FindOrder}(\ell, L)$
\For {\textbf{each} $h(\lambda)\in L$}
\State $r_0\gets \lambda \bmod h(\lambda) \in \F_q[\lambda]/(h(\lambda))$, $r_i\gets r_{i-1}^\ell$ for $i=1,\dots, e-1$
\State Let $f_i(\lambda)$ be the minimal polynomial of $r_i$ over $\F_q$ for $i=1,\dots, e-1$
\State Add $f_1(\lambda),\dots,f_{e-1}(\lambda)$ to $L$
\EndFor
\EndFor
\State \Return $L$
\end{algorithmic}
\end{algorithm}

The procedure  $\mathtt{FindCyclotomic}(g_0(\lambda), n)$ maintains a subset $L$ of irreducible factors of $\Phi_k(\lambda)$ associated with some subgroup of $G$ containing $G_0$.
Initially $L=\{g_0(\lambda)\}$, associated with $H_0:=G_0$. We claim:
\begin{lemma}\label{lem_gen}
Suppose $L$ is associated with $H_{i-1}$ at the beginning $i$th execution of the outer loop of $\mathtt{FindCyclotomic}(g_0(\lambda), n)$. Then at the end of the  $i$th execution, the set $L$ is associated with  a subgroup $H_i\supseteq H_{i-1}$.
Moreover, $H_i=H_{i-1}$ if $\bar{\ell}\not\in G$ in the $i$th execution of the outer loop. Otherwise $H_i=H_{i-1}\langle\bar{\ell}\rangle$.
\end{lemma}

\begin{proof}
If $\ell\not\in G$  in the $i$th execution of the outer loop, then $e$ is set to zero by Lemma \ref{lem_findorder} and the claim is trivial. So assume $\ell\in G$ and let $H=H_{i-1}\langle\bar{\ell}\rangle$. Then $e$ is the order of $H_{i-1}\bar{\ell}$ in $H/H_{i-1}$ by Lemma \ref{lem_findorder}, or $[H: H_{i-1}]$.
Suppose the irreducible factors in $L$ at the beginning of the $i$th execution correspond to distinct $G_0$-orbits  $G_0\theta_1,\dots, G_0\theta_m$ whose union is the $H_{i-1}$-orbit $H_{i-1}\theta$, $m=[H_{i-1}:G_0]$.
The inner loop enumerates $G_0 \theta_j$, and for each of them, adds the irreducible factor corresponding to $G_0 \theta_j^{\ell^s}$ to $L$, $s=1,\dots,e-1$. Note that the union of these $G_0$-orbits $G_0 \theta_j^{\ell^s}=G_0\bar{\ell}^s \theta_j=\bar{\ell}^s G_0 \theta_j$ where $1\leq j\leq m$, $0\leq s\leq e-1$ equals the union of $H_{i-1}$-orbits $\bar{\ell}^s H_{i-1} \theta_j$, which equals the $H$-orbit $H \theta$. And these $G_0$-orbits are all distinct since the number of them is $me=[H:G_0]$. So $L$ is associated with $H$  at the end  of the $i$th execution of the outer loop.
\end{proof}

\begin{lemma}\label{lem_findcyc}
The procedure $\mathtt{FindCyclotomic}(g_0(\lambda), n)$ returns a set $L$ associated with $H_N\subseteq G$. And $H_N=G$ with probability $1-\mathrm{poly}(n)$  in which case $L$ contains  all irreducible factors of $\Phi_k(\lambda)$ over $\F_q$. Moreover $\mathtt{FindCyclotomic}(g_0(\lambda), n)$ could be implemented in time $\widetilde{O}(\phi(k)\log q)$.
\end{lemma}

\begin{proof}
We want to show $H_N=G$ with probability $1-\mathrm{poly}(n)$. By Lemma \ref{lem_findorder} and  Lemma \ref{lem_gen}, it is equivalent to showing the set of $\bar{\ell}\in G$ generates $G$.
Identify $G$ with a product of at most $\log |G|\leq \log n$ primary cyclic groups $C_i$ whose orders are coprime to each other. We only need to show the the set of holomorphic images of $\bar{\ell}\in G$ generates $C_i$ for each $i$  with probability $1-\mathrm{poly}(n)$ and then apply the union bound.

So fix one such $C_i$ and let $m=|C_i|$. Then $\phi(m)$ out of the $m$ elements in $C_i$ are generators of $C_i$. Let $\alpha$ be the probability that the holomorphic image of $\bar{\ell}$ is among these $\phi(m)$ elements, where $\ell$ is randomly sampled from $[1,n]$ as on Line 3. As $m$ is a prime power, we have $\phi(m)\geq m/2$. Therefore
$$
\alpha \geq \frac{\lfloor n/k\rfloor}{n}\cdot \frac{\phi(m)}{m}\cdot |G|=\Omega(\phi(k)/k)=\Omega(1/\log\log k)
$$
where we use $k/\phi(k)=O(\log\log k)$ \cite{rs}.
So for sufficiently large $N=\lfloor c \log n\log\log n\rfloor$, the claim holds with probability $1-\mathrm{poly}(n)$.

Then we analyze the running time:  Line 4 runs in time $\widetilde{O}(\phi(k)\log q)$ by Lemma  \ref{lem_findorder}. Line 7 could be implemented in time $\widetilde{O}(|G_0|\log q)$ \cite{ku,sho}. And Line 3--9 runs in time $|L|\cdot \max\{e, 1\}\cdot \widetilde{O}(|G_0|\log q)$. This is bounded by $\widetilde{O}(\phi(k)\log q)$ since $|L|=[H_{i-1} : G_0]$, $\max\{e,1\}=[H_i:H_{i-1}]$, and $|G|=\phi(k)$. As $N=\Theta(\log n\log\log n)$,  the total running time is bounded by  $\widetilde{O}(\phi(k)\log q)$.
\end{proof}

\subsubsection{Finding the integer $k$}

Another problem we need to solve is finding the integer $k$ given an irreducible factor $g_0(\lambda)$ of $\Phi_k(\lambda)$ over $\F_q$.
Using the procedure $\mathtt{FindCyclotomic}(g_0(\lambda))$, we could find all the irreducible factors of $\Phi_k(\lambda)$ and hence $\Phi_k(\lambda)$ itself.  The degree $d:=\deg(\Phi_k(\lambda))=\phi(k)$ is hence also known.
If $|\phi^{-1}(d)|$ is small, we could find $k$ by enumerating $k_0\in\phi^{-1}(d)$ and checking if $\Phi_{k_0}(\lambda)=\Phi_k(\lambda)$.  However, Erd\H{o}s \cite{erd} showed that for some constant $c>0$, there are infinitely many integers $d$ for which $|\phi^{-1}(d)|\geq d^c$. So this approach is not affordable in general. 
Instead, we use the following procedure to find $k$ efficiently:

\begin{algorithm}[H]
\caption{$\mathtt{Findk}(d,g_0(\lambda))$}
\begin{algorithmic}[1]
\INPUT An integer $d|\phi(k)$ and an irreducible factor $g_0(\lambda)$ of $\Phi_k(\lambda)$ over $\F_q$
\OUTPUT A positive integer $k_0| k$. And $k_0=k$ if $d=\phi(k)$
\State $k_0\gets 1$
\For {\textbf{each} prime $\ell$ such that $(\ell-1)|d$}
\State $e\gets 0$,  $h(\lambda)\gets g_0(\lambda)$, and $r=\lambda \bmod g_0(\lambda) \in \F_q[\lambda]/(g_0(\lambda))$
\While{$\mathtt{FindOrder}(\ell, \{h(\lambda)\})=0$}
\State $e\gets e+1$
\State $r\gets r^\ell$
\State Let $h(\lambda)$ be the minimal polynomial of $r$ over $\F_q$
\EndWhile
\State $k_0\gets \ell^e\cdot k_0$
\EndFor
\end{algorithmic}
\end{algorithm}

\begin{lemma}\label{lem_findk}
There exists a procedure $\mathtt{Findk}(d,g_0(\lambda))$ that takes an integer $d>0$ dividing $\phi(k)$ and 
an irreducible factor $g_0(\lambda)$ of $\Phi_k(\lambda)$, and returns a positive integer $k_0|k$  in time $\widetilde{O}(\phi(k)\log q)$. Moreover $k_0=k$ if $d=\phi(k)$.
\end{lemma}

\begin{proof}
First assume $d=\phi(k)$. We compute $k_0=k$ by determining its prime divisors $\ell$ and $v_\ell(k)$.
Note that if a prime integer $\ell$ divides $k$, we have $(\ell-1) | \phi(k)$.
So we enumerate all primes $\ell$ for which $(\ell-1) | d$ as on Line 2. Then Line 3--8 determines  $e=v_\ell(k)$.
To do this, we start with $e=0$ and keep increasing $e$ until $\ell$ is not invertible in $\Z/(k/\ell^e)\Z$,  or equivalently $e=v_\ell(k)$, i.e., $e$ is the integer satisfying $\ell^e|k$ and $\ell^{e+1}\nmid k$.
To check if $\ell$ is not invertible in $\Z/(k/\ell^e)\Z$, we maintain $r$ as a primitive $(k/\ell^e)$th root of unity and $h(\lambda)$ its minimal polynomial over $\F_q$. Applying Lemma \ref{lem_findorder} with $k'=k/\ell^e$ in place of $k$, we see $\ell$ is not invertible in $\Z/(k/\ell^e)\Z$ if and only if $\mathtt{FindOrder}(\ell, \{h(\lambda)\})$ returns 0, which is checked on Line 4.
Finally, $k_0$ is the product of all $\ell^e=\ell^{v_\ell(k)}$ at the end of the procedure and hence equals $k$.\\ \\
On the other hand,  if $d$ is only a proper divisor of $\phi(k)$, not necessarily all prime divisors $\ell$ of $k$ are enumerated. But we still have $k_0|k$ by the argument above.\\ \\
For the running time, note that the number of $\ell$ we enumerate is bounded by the number of divisors of $d\leq k$, which is bounded by $k^{O(1/\log\log k)}=\phi(k)^{o(1)}$ by a classical result of Wigert \cite{wig}. The inner loop is executed at most $\log k$ times for each $\ell$. The condition on Line 4 can be checked in time  $\widetilde{O}(\phi(k)\log q)$  by Lemma \ref{lem_findorder}  and Line 7 can also be implemented in time $\widetilde{O}(\phi(k)\log q)$ \cite{ku,sho}. The claim follows.
\end{proof}

\subsubsection{Finding the set $T$}

Now we are ready to describe the procedure $\mathtt{FindT}(g_1(\lambda), \dots, g_m(\lambda))$:

\begin{algorithm}[H]
\caption{$\mathtt{FindT}(g_1(\lambda), \dots, g_m(\lambda))$}
\begin{algorithmic}[1]
\INPUT The irreducible factors $g_1(\lambda), \dots, g_m(\lambda)$ of $g(\lambda)$ over $\F_q$
\OUTPUT The set $T$ and multiplicities $m_k$ for each $k\in T$
\State $L_0\gets \{g_1(\lambda), \dots, g_m(\lambda)\}$ as a multi-set and $T\gets \emptyset$
\Repeat
\State Pick an arbitrary element $g_0(\lambda)\in L_0$
\State $L\gets\mathtt{FindCyclotomic}(g_0(\lambda))$
\State $h(\lambda)\gets \prod_{f_i(\lambda)\in L} f_i(\lambda)$, $d\gets \deg(h(\lambda))$
\State $k_0\gets \mathtt{Findk}(d,g_0)$
\If {$h(\lambda)|\lambda^{k_0}-1$}
\State \textbf{if} $k_0\not\in T$ \text{then} $m_{k_0}\gets 0$
\State $T\gets T\cup \{k_0\}$, $m_{k_0}\gets m_{k_0}+1$
\State $L_0\gets L_0-L$
\EndIf
\Until{$L_0=\emptyset$}
\State \Return $T$
\end{algorithmic}
\end{algorithm}

\begin{theorem}[Theorem \ref{thm_timeT} restated]
$\mathtt{FindT}(g_1(\lambda),\dots,g_m(\lambda))$ computes the set $T$ and multiplicities $m_k$  as defined in Algorithm \ref{algo2}, Line $4$. Moreover it halts in time $\widetilde{O}(n \log q)$  with probability $1-1/\mathrm{poly}(n)$.  
\end{theorem}

\begin{proof}
The algorithm picks $g_0(\lambda)$ from $L_0$, calls $\mathtt{FindCyclotomic}$ to find a list $L\subseteq L_0$ that almost surely contains all the irreducible factors of $\Phi_k(\lambda)$, and remove these factors from $L_0$. It repeats these steps until $L_0$ is empty. Each time it also determines the integer $k$ using $\mathtt{Findk}$, adds it to $T$ and updates $m_k$.\\ \\
Note that with small probability, the list $L$ returned by $\mathtt{FindCyclotomic}$  may not contain all the irreducible factors, in which case it is associated with a proper subgroup $H_N\subseteq G$ (c.f. Lemma \ref{lem_findcyc}). In any case we have $\deg(h(\lambda))|\phi(k)$ and therefore by Lemma \ref{lem_findk}, the integer $k_0$ returned by $\mathtt{Findk}$ divides $k$.
We verify that $k=k_0$ on Line 7: $h(\lambda) | (\lambda^{k_0}-1)$ if and only if $k|k_0$ if and only if $k=k_0$ since we know $k_0|k$. And if we find $k\neq k_0$ we do nothing in that round. The correctness of the algorithm is then straightforward. \\ \\
For the running time, note that each round runs in time $\widetilde{O}(\phi(k)\log q)$ by Lemma \ref{lem_findcyc} and Lemma \ref{lem_findk}, and  
then factors of total degree $\phi(k)$ are removed from $L_0$ with probability $1-1/\mathrm{poly}(n)$. So  with probability $1-1/\mathrm{poly}(n)$, the total running time is bounded by $\sum_{i=1}^m  \widetilde{O}(\deg(g_i(\lambda)) \log q)=\widetilde{O}(n \log q)$.
\end{proof}

\section{Polynomial Factorization Using Carlitz Modules}\label{section_carlitz}
We next establish connections between polynomial factorization and the Carlitz action. We prove two nearly linear reductions, namely \textsc{Factor Degree} to \textsc{Carlitz Char-Poly} and  \textsc{Carlitz Char-Poly} to \textsc{Factor}. The former reduction requires that the characteristic $p$ of $\F_q$ is larger than the number of irreducible factors. 
\subsection{Carlitz Modules}
Let $A$ be an $\F_q[x]$-algebra. For $f(x) \in \F_q[x]$ and $\alpha \in A$, $f(x)\alpha$ is understood to be the result of the $\F_q[x]$ action of $f(x)$ on $\alpha$ in $A$. Let $\sigma : A \longrightarrow A$ and $\tau : A \longrightarrow A$ denote the $q^{th}$ power Frobenius endomorphism and the multiplication by $x$ endomorphism respectively. That is, $\forall \alpha \in A$, $\sigma(\alpha)=\alpha^q$ and $\tau(\alpha)=x\alpha$. In \cite{car,car1}, Carlitz endowed a new $\F_q[x]$-module structure on $A$ by defining $m(x) = \sum_i m_i x^i \in \F_q[x]$ to act on $\alpha \in A$ as $$\rho_m(\alpha):= \left(m(\sigma + \tau)\right)(\alpha) = \left(\sum_i m_i (\sigma+\tau)^i \right)(\alpha).$$
In particular, $\forall \alpha \in A$, $\rho_x(\alpha) = \alpha^q+x\alpha$ and $\forall u \in \F_q$, $\rho_u(\alpha) = u\alpha$. 
Let $\rho(A)$ denote the $\F_q[x]$-module structure thus endowed to $A$ by the Carlitz action. To factor a monic square free polynomial $f(x)$, we will concern ourselves with $\rho(\F_q[x]/(f(x)))$. Let $\chi_f(x) \in \F_q[x]$ denote the characteristic polynomial of the $\F_q$ linear transformation on $\F_q[x]/(f(x))$ that takes $\alpha \in \F_q[x]/(f(x))$ to $\rho_x(\alpha)$. Hence \textsc{Carlitz Char-Poly} may be restated as
\begin{problem}\label{carlitz_problem} Given a monic square free $f(x) \in \F_q[x]$, compute $\chi_f(x)$.
\end{problem}
By Lemma \ref{carlitz_crt_lemma}, knowledge of factorization of $f(x)$ immediately yields $\chi_f(x)$ in $\widetilde{O}(n \log q)$ time. Thus Problem \textsc{Carlitz Char-Poly} is linear time reducible to \textsc{Factor}. We next reduce \textsc{Factor Degree} to \textsc{Carlitz Char-Poly}.
\subsection{Factor Degree Estimation using Carlitz Modules}\label{carlitz_crt_section}
\begin{lemma}\label{carlitz_crt_lemma}
Let $f(x) = \prod_if_i(x)$ be a factorization of a monic square free $f(x) \in \F_q[x]$ into monic irreducible polynomials . Then $ \rho(\F_q[x]/(f(x)))\cong \bigoplus_i \F_q[x]/(f_i(x)-1).$ In particular, $\chi_f(x) = \prod_i(f_i(x) -1)$.
\end{lemma}
\begin{proof}
By the Chinese remainder theorem, $\F_q[x]/(f(x)) \cong \prod_i \F_q[x]/(f_i(x))$
\begin{equation}\label{carlitz_crt_equation}
\Rightarrow \rho(\F_q[x]/(f(x))) \cong \rho\left(\prod_i \F_q[x]/(f_i(x))\right) \cong \bigoplus_i \rho\left(\F_q[x]/(f_i(x))\right).
\end{equation}
The final congruence holds since for every direct product $C \cong A \times B$ of $\F_q[x]$-algebras, we have the corresponding direct sum $\rho(C) \cong \rho(A) \oplus \rho(B)$ of $\F_q[x]$-modules \cite{kcon}. For a monic irreducible $g(x)$ (\cite{kcon}), 
\begin{equation}\label{carlitz_irred}
\rho(\F_q[x]/(g(x))) \cong \F_q[x]/(g(x)-1).
\end{equation}
Equation \ref{carlitz_crt_equation} and \ref{carlitz_irred} together prove the lemma.  
\end{proof}
\begin{lemma}\label{factor_degree} If $p$ does not divide the number of smallest degree factors of a monic square free $f(x) \in \F_q[x]$, then the smallest irreducible factor degree of $f(x)$ is $\deg(f(x)) - \deg(f(x) -  \chi_f(x))$.
\end{lemma}
\begin{proof}  Let $f(x) = \prod_{i}f_i(x)$ be a factorization of a monic square free $f(x) \in \F_q[x]$ into monic irreducible polynomials.
Let $d$ be the smallest degree of factors of $f(x)$.
Then  
$$f(x)-\chi_f(x) = f(x) - \prod_i (f_i(x)-1) = \sum_{i} \frac{{f(x)}}{f_i(x)} + (\text{terms of degree less than } \deg(f(x))-d).$$
The first equality is from Lemma \ref{carlitz_crt_lemma}. Since $f(x)$ and $f_i(x)$ are all monic and $p$ does not divide the number of  $f_i(x)$ of degree $d$,
the leading term of $\sum_{i} ({f(x)}/f_i(x))$ is of degree $\deg(f(x))-d$. Therefore $\deg(f(x)-\chi_f(x))=\deg(f(x))-d$ and the lemma follows.
\end{proof}
\noindent Lemma \ref{factor_degree} reduces in nearly linear time \textsc{Factor Degree} (when restricted to $p$ greater than the number of factors of $f(x)$) to \textsc{Carlitz Char-Poly}. To see this, given $f(x)$, we may call an algorithm that solves Problem \ref{carlitz_problem} to obtain $\chi_f(x)$ and output $\deg(f(x)) - \deg(f(x)-\chi_f(x))$.

\section{Moore and Vandermonde Determinants}\label{section_moore}
\subsection{Moore Determinants and Carlitz Factorials}
Let $A$ be a finitely generated $\F_q$ algebra and $n$ a positive integer. The Moore matrix $M_w$ with first row $w=(w_1,w_2,\ldots,w_n) \in A^n$ is defined as 
\[M_w:=
\begin{bmatrix}
    w_1 & w_2 & w_3 & \dots  & w_n \\
    w_1^q & w_2^q & w_3^q & \dots  & w_n^q \\
    w_1^{q^2} & w_2^{q^2} & w_3^{q^2} & \dots  & w_n^{q^2} \\
    \vdots & \vdots & \vdots & \ddots & \vdots \\
    w_1^{q^{n-1}} & w_2^{q^{n-1}} & w_3^{q^{n-1}} & \dots  & w_n^{q^{n-1}}
\end{bmatrix}
\] 
and its determinant $det(M_w)$ is denoted by $\Delta(w_1,w_2,\ldots,w_n)$. 
For a positive integer $m$, the $m^{th}$ Carlitz factorial 
$$\prod_{0\leq i<j\leq m}\left(x^{q^{j-i}}-x\right)^{q^i},$$ is the product of all polynomials over $\F_q$ of degree at most $m$ \cite{car}. We next recall Carlitz's identity and from it reduce \textsc{Factor Degree} to computing certain Moore determinants.
\begin{lemma}\label{carlitz_factorial} (\textit{Carlitz \cite{car}}) For every positive integer $m$,
$$\Delta(1,x,x^2,\ldots,x^m) = \prod_{0\leq i<j\leq m}\left(x^{q^{j-i}}-x\right)^{q^i},$$
\end{lemma}
\begin{proof}
The Moore matrix with first row $(1,x,x^2,\ldots,x^m)$, when viewed column-wise is Vandermonde. By the Vandermonde determinant formula,
\[
det\left(\begin{bmatrix}
    1 & x & x^2 & \dots  & x^{m} \\
    1 & x^q & x^{2q} & \dots  & x^{mq} \\
    1 & x^{q^2} & x^{2q^2} & \dots  & x^{mq^2} \\
    \vdots & \vdots & \vdots & \ddots & \vdots \\
    1 & x^{q^{n-1}} & x^{2q^{n-1}} & \dots  & x^{mq^{n-1}}
\end{bmatrix}\right) = \prod_{0\leq i<j \leq m} \left(x^{q^j}-x^{q^i}\right) = \prod_{0\leq i<j\leq m}\left(x^{q^{j-i}}-x\right)^{q^i}
\] 
\end{proof}
\textsc{Moore-Det} may be restated as
\begin{problem}\label{moore_det_problem} Given a square free monic polynomial $f(x) \in \F_q[x]$ of degree $n$ and a positive integer $m\leq n$, decide if $\Delta(1, x,\ldots, x^{m}) \mod f(x)$ is $0$.
\end{problem}
Problem \ref{moore_det_problem} can be solved in $\widetilde{O}(n^{3/2} \log q+n\log^2 q)$ time \cite[Lemma~8.4]{ku}.
\begin{theorem}\label{moore_zero_theorem}
If there is a $T(n,m,\log q)$ time algorithm for Problem \ref{moore_det_problem}, then \textsc{Factor Degree} can be solved in $O(T(n,\lceil n/2\rceil, \log q) \log n)$ time. That is, \textsc{Factor Degree} is nearly linear time reducible to \textsc{Moore-Det}.
\end{theorem}
\begin{proof}
By Lemma \ref{carlitz_factorial}, for a monic square free $f(x)\in \F_q[x]$ and $m \leq \deg(f(x))$, we have $\Delta(1, x,\ldots, x^{m}) \mod f(x)=0$ if and only if 
\begin{equation}\label{carlitz_factorial_equation}
\prod_{0\leq i<j\leq m}\left(x^{q^{j-i}}-x\right)^{q^i} = 0 \mod f(x).
\end{equation}
Since $f(x)$ is square free, Equation \ref{carlitz_factorial_equation} holds if and only if every irreducible factor of $f(x)$ has degree at most $m$. Given oracle access to an algorithm for Problem \ref{moore_det_problem}, a binary search leads to the determination of the largest irreducible factor degree of $f(x)$. 
\end{proof}
\subsection{Vandermonde Determinants}
The determinants involved in the previous subsection were both Moore and Vandermonde. Here we study determinants that are Vandermonde but not Moore. Further, the matrices involved are of dimension significantly smaller than the degree of the polynomial factored.

For a positive integer $m$, let $$S_m:=\{0,1,2,\ldots,\lfloor\sqrt{m}\rfloor-1,\lfloor\sqrt{m}\rfloor,2\lfloor\sqrt{m}\rfloor,3\lfloor\sqrt{m}\rfloor,\ldots,(\lfloor\sqrt{m}\rfloor-1)\lfloor\sqrt{m}\rfloor,\lfloor\sqrt{m}\rfloor^2,m\}.$$
This ensures that $|S_m|\leq 2\lfloor\sqrt{m}\rfloor+1$ and $\{j-i|i,j \in S_m,i<j\} =\{1,2,\ldots,m-1,m\}$.

For a positive integer $m$, let $V_{m}(x) \in \F_q[x]$ denote the determinant of the Vandermonde matrix with first row $\{x^{q^i},i\in S_m\}$. 
\begin{lemma}\label{vandemonde_carlitz}
For every monic square free $f(x)\in \F_q[x]$ and every positive integer $m$,
$$\gcd(V_{m}(x),f(x)) = \gcd\left(\prod_{0\leq i \leq m}\left(x^{q^{i}}-x\right), f(x)\right).$$
\end{lemma}
\begin{proof}
By the Vandermonde determinant formula, 
\begin{equation}\label{vandermonde_equation1}
  V_{m}(x) = \prod_{i,j \in S_m|i<j}\left(x^{q^j}-x^{q^i}\right) =  \prod_{i,j \in S_m|i<j}\left(x^{q^{j-i}}-x\right)^{q^i}.
\end{equation}
Since $f(x)$ is square free and $\{j-i|i,j \in S_m,i<j\} =\{1,2,\ldots,m-1,m\}$, 
\begin{equation}\label{vandermonde_equation2}
\gcd\left(\prod_{i,j \in S_m|i<j}\left(x^{q^{j-i}}-x\right)^{q^i}, f(x)\right) = \gcd\left(\prod_{0\leq i \leq m}\left(x^{q^{i}}-x\right), f(x)\right).
\end{equation}
By Equations \ref{vandermonde_equation1} and \ref{vandermonde_equation2}, the lemma follows.
\end{proof}
\textsc{Vandermonde Det} may be restated as
\begin{problem}\label{vandermonde_det_problem} Given a square free monic polynomial $f(x) \in \F_q[x]$ of degree $n$ and a positive integer $m\leq n$, decide if $V_{m}(x) \mod f(x) $ is $0$.
\end{problem}
We next sketch a fast algorithm for Problem \ref{vandermonde_det_problem}. Since $|S_m| \leq \sqrt{n}$, the first row $\{x^{q^i} \mod f(x),i \in S_m\}$ can be computed in $\widetilde{O}(n^{3/2} \log q+n\log^2 q)$ time using iterated Frobenius algorithm \cite{gs} implemented using fast modular composition  \cite{ku}. Given the first row of a Vandermonde matrix over a commutative ring, the square of its determinant can be computed with nearly linearly many operations over the ring \cite{pan}. Hence, $V_{m}(x) \mod f(x)$ can be zero tested in $\widetilde{O}(n^{3/2} \log q+n\log^2 q)$ time.
\begin{theorem}\label{vandermonde_reduction_theorem}
If there is a $T(n,m,\log q)$ time algorithm for Problem \ref{vandermonde_det_problem}, then \textsc{Factor Degree} can be solved in $O(T(n,\lceil\sqrt{n}\rceil, \log q) \log n)$ time. That is, \textsc{Factor Degree} is nearly linear time reducible to \textsc{Vandermonde Det}.
\end{theorem}
\begin{proof}
By Lemma \ref{carlitz_factorial} and Lemma \ref{vandemonde_carlitz}, for every monic square free $f(x)\in \F_q[x]$ and positive integer $m\leq \deg(f(x))$, $$\gcd(V_{m}(x),f(x)) = \gcd(\Delta(1, x,\ldots, x^{m}),f(x)).$$ 
Hence Problems \ref{moore_det_problem} and  \ref{vandermonde_det_problem} are identical and our theorem follows from Theorem \ref{moore_zero_theorem}.
\end{proof}
\begin{remark}  
We may pose functional variants of Problems \ref{moore_det_problem} and \ref{vandermonde_det_problem}, by asking for the respective determinants module $f(x)$, instead of merely deciding if they are zero.
\begin{problem}\label{moore_det_problem_functional} Given a square free monic polynomial $f(x) \in \F_q[x]$ of degree $n$ and a positive integer $m\leq n$, compute $\Delta(1, x,\ldots, x^{m}) \mod f(x)$.
\end{problem}
\begin{problem}\label{vandermonde_det_problem_functional} Given a square free monic polynomial $f(x) \in \F_q[x]$ of degree $n$ and a positive integer $m\leq n$, compute $(V_{m}(x))^i \mod f(x)$ for some positive integer $i$.\end{problem}
As outlined, $3/2$ exponent algorithms are known for problems \ref{moore_det_problem_functional} and \ref{vandermonde_det_problem_functional}. By \cite[Thm~8.5]{ku},  \textsc{Factor} is nearly linear time reducible to each of the problems \ref{moore_det_problem_functional} and \ref{vandermonde_det_problem_functional}.  
\end{remark}

\bibliographystyle{plain}
\bibliography{ref}

\end{document}